\documentclass[10pt]{article}

\usepackage{graphicx,amsmath,amsthm,amsfonts,setspace}

\newcommand{\diag}{\text{diag}}

\newcommand{\mn}{\mathbb{N}} 
\newcommand{\bC}{\mathbf{C}}

\newcommand{\bS}{\mathbf{S}}

\newcommand{\bD}{\mathbf{D}}
\newcommand{\bR}{\mathbf{R}}
\newcommand{\br}{\mathbf{r}}
\newcommand{\bU}{\mathbf{U}}

\newcommand{\sinc}{\text{sinc}}
\newcommand{\ebno}{E_b/N_0}

\newtheorem{theorem}{Theorem}

\author{Richard G. Clegg, Safa Isam, Ioannis Kanaras and 
Izzat Darwazeh}

\title{A practical system for improved efficiency in frequency 
division multiplexed wireless networks}

\begin{document}

\setcounter{page}{0}
\maketitle
\thispagestyle{empty} 

\begin{abstract}
Spectral efficiency is a key design issue for all wireless communication systems.  
Orthogonal frequency division multiplexing (OFDM) is a very well-known technique for efficient 
data transmission over many carriers overlapped in frequency.  Recently, several 
papers have appeared which describe 
spectrally efficient variations of multi-carrier systems where the condition 
of orthogonality is dropped.  Proposed techniques suffer from two weaknesses:  
Firstly, the complexity of  generating the signal is increased. Secondly, the 
signal detection 
is computationally demanding.  Known methods suffer either unusably high complexity
or high error rates because of the inter-carrier interference.  
This work addresses both problems by proposing new transmitter and receiver architectures whose design 
is based on using the simplification that a rational Spectrally Efficient 
Frequency Division Multiplexing (SEFDM) 
system can be treated as a set of overlapped and interleaving 
OFDM systems. 

The efficacy of the proposed designs is shown through detailed simulation of systems with different 
signal types and carrier dimensions.  The 
decoder is heuristic but in practice produces very good results which are close to the theoretical 
best performance in a variety of settings.  The system is able to produce 
efficiency gains of up to 20\% with negligible impact on the required
signal to noise ratio.
\end{abstract}


\pagebreak

\section{Introduction}
\label{sec:intro}
Orthogonal Frequency Division Multiplexing (OFDM) is a well-known technique
for efficient data transmission. OFDM is at the core of
communications technologies such as Digital Audio Broadcasting (DAB) and 
Digital Video Broadcast
(DVB), wireless broadband networks such as 
Worldwide Interoperability for Microwave Access
(WiMAX) and long 
term evolution (LTE) systems. 
In OFDM, data is transmitted using a number of orthogonal carrier frequencies.
Recently many authors have proposed non-orthogonal systems or
Spectrally Efficient Frequency Division Multiplexing (SEFDM) systems.
OFDM symbols are sent on carrier frequencies separated by
$F$ and the symbols remain constant for time $T$ (the symbol
period) with $TF=1$.  This ensures no sub-channel interference.
For SEFDM, $TF = \alpha < 1$ and, while there will necessarily 
be sub-channel interference, the key advantage is that the available 
spectrum can be used more 
efficiently.

This paper suggests design for a simple to implement transmitter and
receiver/decoder for SEFDM systems.  The transmitter design and the decoder design
are interlinked.  The key
insight is to see SEFDM as a small number of interleaved OFDM systems.
The design can increase spectral efficiency by 20\%
using similar techniques to traditional OFDM and with little
compromise to the required signal to noise ratio for the system.  
The designs require only slightly more complexity in the receiver and
transmitter. 
The decoder
design is not optimal but is, instead, designed to be fixed (and low) complexity
heuristic
with ``good enough" performance.  It is shown in simulation that gains significantly 
above 20\%
are unlikely without a more radical redesign of SEFDM systems since even optimal
decoding begins to suffer large increases in bit error rate 
relative to OFDM.

Section \ref{sec:back}
describes other research and the background.
The structure of the paper is as follows. Section \ref{sec:sefdm_intro} 
provides a brief introduction to SEFDM.   Section 
\ref{sec:sefdmmaths} derives the main theorem necessary for
our receiver and decoder design.  Section \ref{sec:design}
describes the receiver and decoder design.
Section \ref{sec:results} shows the designs perform well in simulation
and section \ref{sec:conclusions} gives conclusions and further work.

\subsection{Spectrally Efficient FDM approaches}
\label{sec:back}

The idea of non orthogonal and spectrally efficient systems occurs in the 1975 work 
of Mazo et al \cite{Mazo}.  More recently, the idea of multi-carrier 
spectrally efficient systems was introduced in
\cite{Rodrigues_GS_FDM} and termed SEFDM.
Similar systems use the names 
high compaction multi-carrier
modulation (HC-MCM) \cite{Hamamura_EfficientMulticarrier,Hamamura2005}
or
overlapped FDM (OvFDM) \cite{Jian2007a}. 
Related systems are fast OFDM (FOFDM) \cite{Rodrigues_FOFDM} and 
the $M$-ary amplitude shift keying OFDM \cite{Xiong_MASK_FOFDM}, 
both proposing reducing the spectrum to the half of an equivalent 
OFDM but subject to the limitation that 
the information symbols are only one-dimensional (e.g. BPSK or ASK). 
In addition, offset QAM proposed in \cite{Hirosaki1981} succeeded 
in eliminating guard bands and hence supported higher spectral efficiency.

Recently the concept of non-orthogonal carriers has found 
its way into the very high bit rate optical communications field. 
The applicability of Fast OFDM concept of \cite{Rodrigues_FOFDM} has 
been demonstrated in \cite{Zhao2010} in a system termed 
Optical Fast OFDM that provides attractive error performance for 
one dimensional modulation schemes. Furthermore,  \cite{Yamamoto2010} 
proposed the so called optical Dense OFDM (DOFDM) which can 
accommodate higher order modulations. 
Simulation and experimental tests confirmed almost the same error 
performance as conventional OFDM. By orthogonally polarizing the 
sub-carriers it is possible to enhance immunity to the 
chromatic dispersion for both conventional OFDM and DOFDM. 
A related
system termed non-orthogonal FDM (NOFDM) 
proposed restoration of orthogonality from the view point of the 
input symbols by employing orthogonal pulse-shaping \cite{Kozek1998}, 
where details of appropriate pulse shapes and power and bit loading 
provided in \cite{Strohmer_OptimalOFDM_Sphere,Strohmer_OptimalOFDM} and 
\cite{Kliks2009}.

There are two problems with SEFDM systems: 
efficiently generating such a signal (the transmitter problem) and
efficiently detecting and decoding 
such a signal. For the transmitter problem, 
a known method (first proposed by the authors)
is to use the inverse 
fractional Fourier transform \cite{Kanaras_GSD}. 
The HC-MCM system shortens the symbol transmission time and hence
transmits by using OFDM techniques with zero-padded input and
truncated output
\cite{Hamamura_EfficientMulticarrier}. 
Recently several techniques to generate SEFDM signals using
the Inverse 
Discrete Fourier Transform (IDFT) have been proposed by the 
authors \cite{safa200909,Safa_CSNDSP,Safa201007} and
have been implemented 
in hardware \cite{perrett2011}.

Obtaining optimal
solutions for the decoding problem is non polynomial (NP) hard.  Various
techniques are suggested: some of the better known solutions
to the decoding problem are zero forcing (ZF) 
\cite{Pammer_ZF-ML_Combination,Kanaras_MMSEML}, 
minimum mean squared error (MMSE) \cite{5:Yuan}, the
sphere decoder \cite{Viterbo_SD_Principle,Kanaras_GSD}
and semi-definite relaxation (SDR)
\cite{Wing-Kin_SDR,Kanaras_SDP}.  Maximum likelihood methods have
extremely high complexity and cannot be used in practice for anything other
than the smallest systems. 
Methods such as SDR, MMSE and ZF have lower complexity  but
introduce a significant error penalty, particularly when noise levels are
high or the number of carriers large \cite{Kanaras_MMSEML}.  
They are therefore unlikely to prove useful in 
systems with many carriers or practical noise levels.

By contrast, the sphere decoder (SD) is 
a method of dynamic programming that can handle the NP hardness of 
overlapped optimisation 
problems achieving the optimal solution --- 
SD
techniques are investigated by Kanaras et al in 
\cite{Kanaras_GSD,Kanaras_PrunedSD}. 
Much promising research has taken place on the use of SD for
SEFDM.  
Furthermore, \cite{Kanaras2010} developed a new sub-optimal 
SD based detector that uses semi-definite programming to reduce 
the complexity of the SD, whereas \cite{safa201105} and 
\cite{safa201103} proposed the use of a fixed complexity 
sphere decoders (FSD) and then a combination of FSD and the 
truncated singular value decomposition (SVD) to solve the problem of
the variable complexity of the SD whilst still providing attractive 
error performance.
SD suffers from two basic drawbacks which have only been
partially overcome.  It requires the inversion of ill-conditioned
matrices (regularisation helps this problem at the expense
of introducing noise) and its 
complexity is not fixed but is, in general,
worse than polynomial \cite{Vikalo1,Jalden_SD_Complexity}. 
The execution time of SD can worsen considerably with many carriers,
in high noise or with low $\alpha$.
Consequently, a practical implementation could be possible only under very 
specific conditions, for relatively small signal dimensions ($N\leq32$) 
and in high signal to noise ratio (SNR) regimes.  
Therefore, the need remains
for a detector technique which can recover signals well and in a short fixed time.  

In SEFDM, the channel equalisation problem needs consideration.
Work has been done on the problem of accounting for channel effects in
SEFDM systems and \cite{ersi2010} shows that joint detection and equalisation
are possible.

An open question remains to what extent it is even theoretically possible to
recover signals.
For sampled SEFDM systems the Bit Error Rate (BER) as 
a function of energy per bit to noise power spectral density 
ratio ($\ebno$) is not known although Mazo
and Landau famously made pioneering work in this area for single
carrier systems
\cite{Mazo_FasterThanNyquist}.  A later result by Rusek et al
\cite{rusek2005,Rusek_FasterThanNyquist} demonstrates 
that for $\alpha > 0.802$ and 4-QAM 
for optimal detection the BER should be exactly that of
OFDM (although technical differences in the system mean that this
result may not precisely carry over to SEFDM systems as considered
in this work).  
It also demonstrates that this is the ``best possible" value of $\alpha$ in this setting
and for lower
$\alpha$ this the performance will diverge
(see \cite{rusek2005} for full details). 
However,
it would be expected the BER for a ``good enough" decoding system to be ``close"
to the BER for OFDM for $\alpha \in (5/6,1)$ and diverge for
smaller $\alpha$ (especially when $\alpha < 0.802$).

\section{The Spectrally Efficient
Frequency Division Multiplexing scheme}
\label{sec:sefdm_intro}
If spectral efficiency is defined as the bitrate transmitted divided by the
amount of spectrum used (bits/s/Hz) then it can be seen that multiplying
the symbol period $T$ by a factor $\alpha < 1$ but keeping the frequency
separation $F$ the same will increase the spectral efficiency (by increasing
the bitrate) by a factor of approximately $1/\alpha$ for a 
large number of carriers.  Here then we take the spectral efficiency 
of the new system as being $1/\alpha$ and hence $\alpha = 5/6$ means
spectral efficiency of 120\% or a 20\% gain.  The result is the same
(and the system mathematically identical) if $T$ is kept constant
and $F$ reduced.  

Assume
that the system has $N$ carrier frequencies each separated by a frequency
separation $F$.  Let $S_i$ (with $i \in \{0, 1, \ldots, N-1\})$
be the symbol (a complex number chosen from an ``alphabet") 
on carrier $i$ for time $[0,t)$.    
Now, ignoring
the frequency offset of the initial carrier for simplicity, the transmitted signal
($B(t)$ for broadcast signal)
in the period $[0,T]$ is given by
$
B(t) = \sum_{k=0}^{N-1} S_k \exp [ 2 \pi i k t / T ].
$
For OFDM, the interference between frequencies
is zero when the signal is integrated over the symbol period.
The discrete version of this can be considered instead
where $B(t)$ is sampled at $M$ discrete times in the set
$\{0, T/M, 2T/M, \ldots, (M-1)T/M\}$.  
This new series is $U_m$ (with $m \in \{0,1,\ldots,M-1\})$ 
where $U_m = B(Tm/M)$ and \eqref{eqn:sefdm_cont} becomes
$U_m = \sum_{k=0}^{N-1}S_k \exp [ 2 \pi i k (mT/M) T] 
 = \sum_{k=0}^{N-1} S_k \exp [ 2 \pi i km/M]$.
It is this discrete version which is traditionally
used in OFDM transmitters as it can be easily generated
using FFT techniques and then the continuous signal approximated
from this.

Now, consider the SEFDM system where $TF = \alpha <1$.  Further we assume
that $\alpha$ is rational $\alpha=b/c$ with $b,c \in \mn$ (the set of 
natural numbers).  
The equivalent equation for the transmitted signal is given by
\begin{equation}
\label{eqn:sefdm_cont}
B(t) = \sum_{k=0}^{N-1} S_k \exp [ 2 \pi i k t b/cT ],
\end{equation}
where $B(t)$ is the broadcast signal at time $t \in [0,T)$.
The discretely sampled version 
where $U_m = B(Tm/M)$  becomes
\begin{equation}
\label{eqn:sefdm_disc}
U_m =  \sum_{k=0}^{N-1} S_k \exp [ 2 \pi i km b/cM].
\end{equation}
Because of the $b/c$ factor
FFT techniques cannot be used in a straightforward manner to generate the
transmitted wave.  However, section \ref{sec:proof} shows
one way this can be done and section \ref{sec:design} shows one workable
design for a transmitter and decoder based on the insight that the 
SEFDM system with rational $\alpha$ consists of interleaved OFDM systems.

A working SEFDM system will generate and receive a continuous waveform.
(If the transmission is digitally generated as in this case the
continuous wave form would be from a smoothed version of the digital
samples).  A computer simulation is by its nature discrete.  It can
be shown that for the continuous wave form, the interchannel interference
impacting the $m$th channel from the $n$th channel in an SEFDM
system (for $n \neq m$) is given by:
$$
I(n,m)= S_n \left(\sinc[(n-m) \alpha ]/\pi\right) \exp[\pi i (n-m) \alpha],
$$
where $\sinc(x)$ is the normalised sinc function $\sin(\pi x)/x$.
For the discrete simulation the interference term is
$$
I'(n,m)= \frac{S_n \sinc[(n-m)\alpha]}{\sinc[(n-m)\alpha/M]}
\exp[\pi i (n-m) \alpha (M-1)/M].
$$
This can be thought of as the original $I(n,m)$ 
corrupted by a rotation factor $(M-1)/M$ (the origin
of this is the non-centred sampling times) and
a magnification factor $(n-m)\alpha/ M \sin[(n-m)\alpha/M]$.  Both tend
to 1 as $M \rightarrow \infty$ (as would be expected).  In short, 
the discrete simulation
will exaggerate (sometimes greatly) the interfering effects
of the SEFDM carriers.

\section{Mathematics of SEFDM systems}
\label{sec:sefdmmaths}

\begin{figure}
\begin{center}
\includegraphics[width=12cm]{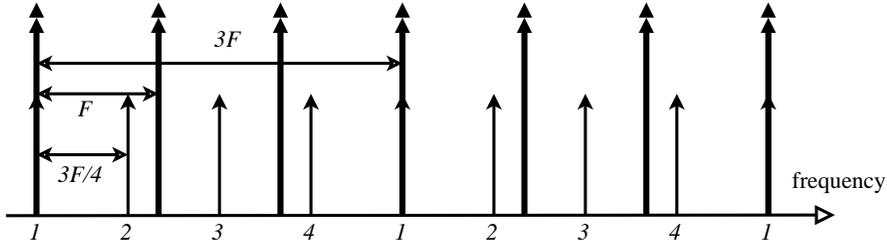}
\caption{A representative diagram of SEFDM with $\alpha=3/4$ compared
to an OFDM system.}
\label{fig:sefdmexample}
\end{center}
\end{figure}

A core insight of this paper is the viewing of SEFDM as interleaved
OFDM systems.  This is illustrated in Fig. \ref{fig:sefdmexample}.
Here the large vertical double arrows represent an OFDM system with
symbol period $T$ and frequency separation $F$.  (Remember that an
OFDM system has $TF=1$ and an SEFDM system has $TF = \alpha < 1$.)
The smaller single
arrows represent an SEFDM system with the same symbol period $T$
and a frequency separation $\frac{3}{4}F$ ($\alpha=3/4$).  

It can be seen that those SEFDM frequencies labelled 1 (below
the x-axis) always align exactly with OFDM frequencies (separated by $3F$).  
In other words, those SEFDM frequencies are an OFDM system which happens 
only to send symbols on every third carrier.  The frequencies 
labelled 2 also form an OFDM system
offset in frequency from the first by $\frac{3}{4}F$.  
In general, if $\alpha$ 
is some rational $\alpha = b/c$ with $b < c \in \mn$ (where 
$\mn$ is the set of natural numbers) this can be
viewed as $c$ interleaved OFDM systems each sending symbols on every
$b$th carrier and offset from each other by $F$.
This insight will be used both in transmitter and decoder design,
a formal proof follows.

\subsection{Proof of the equivalence of SEFDM and interleaved OFDM}
\label{sec:proof}

For notational 
convenience, assume here and throughout this paper that
an $N \times M$ matrix $\bC = [c_{nm}]$ has its indices running from zero
(not one as is more usual).  That is $n \in \{0,1,\ldots,N-1\}$ and 
$m \in \{0, 1, \ldots, M-1\}$.  
Equation \eqref{eqn:sefdm_disc} can now be written as 
$\bU = \bS \bC$, where $\bU = [U_0, U_1, \ldots, U_{M-1}]$, 
$\bS = [S_0, S_1, \ldots, S_{N-1}]$ and $\bC= [c_{nm}]$ is the 
$N \times M$ carrier matrix given by
$
c_{nm} =  \exp [ 2 \pi i nm b/cM].$
The equation $\bU= \bS \bC$ does generate the sequence to be transmitted
but the multiplication by an $N \times M$ matrix could be computationally 
intensive for large $N$.
Assume $N$ is some multiple
of $c$ (this is not a necessary assumption but makes the
notation easier) so $N=cd$ with $d \in \mn$.

\begin{theorem}
\label{thm:sefdm}
Consider an SEFDM system with $N$ carrier frequencies sampled $M \geq N$ 
times in
the symbol period and  with
rational $\alpha = b/c$ as described in equation \eqref{eqn:sefdm_disc}.
The system can be decomposed into the sum of $c$ separate
OFDM systems each with $b\lceil N/c \rceil$ carrier frequencies
and a frequency offset applied to each of these
(where $\lceil \cdot \rceil$ is the ceiling function). 
Of these frequencies, a maximum of $\lceil N/c \rceil$ actually
carry symbols.
\end{theorem}

\begin{proof}
Assume without loss of generality that $N$ is a multiple of $c$.
For a system where $N$ is not a multiple of $c$ the same proof
applies on the expanded system with $N' > N$ carriers 
such that $N'$ is a multiple of $c$ and no symbols are
broadcast on the final carriers ($S_N = S_{N+1} =\cdots = S_{N'-1} = 0)$.

Let $\bD$ be the $Nb/c \times M$ matrix for an OFDM system
with $Nb/c$ carriers and $M$ samples given by
$\bD =  [ d_{nm}]$ and  
$d_{nm} =   \exp [ 2 \pi i nm/M]$.
Let $\bR(k)$ be the $M \times M$ rotation matrix:
$\bR(k)= \diag[r(k)_m]$ where
$r(k)_m = \exp[2 \pi i mkb/cM]$ and $\diag$ is the matrix with
all elements zero apart from the diagonal which has its $k$th
element as $r(k)_m$.

The SEFDM transmission can be considered as the sum of $c$ interleaved
OFDM systems.  Let $\bU'(k)$ be the signal generated by the $k$th such system.
Let $\bS'(k)$ be the symbols in $\bS$ that are transmitted on the $k$th system.
That is $\bS'(0) = (S_0, S_c, S_{2c},\ldots)$ and 
$\bS'(1) = (S_1, S_{c+1}, S_{2c+1},\ldots)$ and so on.  Formally,
define the $c$ symbol vectors (each of length $Nb/c$)
$\bS'(k)$, 
for $k=1,2,\ldots,c-1$
$$
S'(k)_n = 
\begin{cases}
S_{nc/b+k} & n \mod b = 0, \\
0 & \text{otherwise},
\end{cases}
$$
where $n \in (0, 1, \ldots, Nb/c)$.
Note that each of the original symbols $S_n$ appears in exactly one of
the new symbol vectors $\bS'(k)$.  This also means
that the reverse map can be constructed
$S_n = S'(n \mod c)_{b (n-k)/c }.$

Consider the matrix equation 
$\bU' = \sum_{k=0}^{c-1} \bS'(k) \bD\bR(k).$
This is the sum of the $k$ new symbol vectors transformed by an OFDM system
and rotated.  It remains to show that $\bU' = \bU$.  Define each element of
this sum as 
$\bU'(k) = \bS'(k) \bD \bR(k)$ and therefore
$\bU' = \sum_{k=0}^{c-1}\bU'(k)$.
For any $k$ the $m$th element of $\bU'(k)$ (referred to here as 
$U'(k)_m$) is given by
\begin{equation} U'(k)_m =  \exp [2 \pi i mkb/cM] 
\sum_{n=0}^{Nb/c-1} S'(k)_n \exp [ 2 \pi i nm/M].
\label{eqn:uksum}
\end{equation}
Since $S'(k)_n = 0$ if $n \mod b \neq 0$ then the sum index $n$
can be transformed using $l=n/b$ to give
$U'(k)_m  =  \exp [2 \pi i mkb/cM] 
\sum_{l=0}^{N/c-1} S'(k)_{lb} \exp [ 2 \pi i mlb/M].$
Since $S'(k)_{lb} = S_{lc+k}$ for all $l \in \{0, 1, \ldots, N/c-1\}$ then 
$
U'(k)_m  =  
\exp [2 \pi i mkb/cM]
\sum_{l=0}^{N/c-1} S_{lc+k} 
\exp [2 \pi i mlb/M].$
The sum must be transformed again
using $p=lc+k$ and hence $l = (p-k)/c$ to
$
U'(k)_m =  
\exp [2 \pi i mkb/cM] 
\sum_{p=k}^{Nc-c+k} S_{p}  \delta(p \mod c)
\exp [2 \pi i m(pb-kb)/cM],
$
where $\delta(n)$ is the delta function which is
equal to 1 if $n= 0$ and 0 otherwise. 
A final sum transformation gives
$
U'(k)_m =  \sum_{n=0}^{N-1} S_n 
\delta(n+k \mod c)  \exp[2 \pi i n m b/cM].$
Clearly then the final proof arises
$$
\bU' = \sum_{k=0}^{c-1}  \sum_{n=0}^{N-1} S_n \delta(n+k \mod c) 
\exp[2 \pi i n m b/cM]
=  \sum_{n=0}^{N-1} S_n \exp[2 \pi i n m b/cM] 
 = \bU,
$$
where the removal of the sum over $k$ at the second equality 
sign occurs because $n+k \mod c = 0$ is
always true for exactly one value of $k$ for any given value of $n$.
\end{proof}

\section{Transmitter and receiver design}
\label{sec:design}
The transmitter and receiver designs outlined in this section have several
advantages over those in the SEFDM literature.  The receiver and transmitter
designs also have much in common which would help with the cost of building
them.

\subsection{Transmitter design}
\label{sec:transmitter}

The generation of SEFDM signals using
the Inverse 
Discrete Fourier Transform (IDFT) has been proposed by the 
authors 
in \cite{safa200909} and this has led to a recent hardware
implementation \cite{perrett2011}.
As the SEFDM signal can be described as a sum of overlapped independent 
rotated OFDM signals, it can be shown that the SEFDM transmitters can 
be built using OFDM generation techniques. An OFDM signal is efficiently generated 
using the Inverse Discrete Fourier Transform (IDFT) \cite{Weinstein_Windowing}. 


\begin{figure}
\centering
\includegraphics [width=12cm] {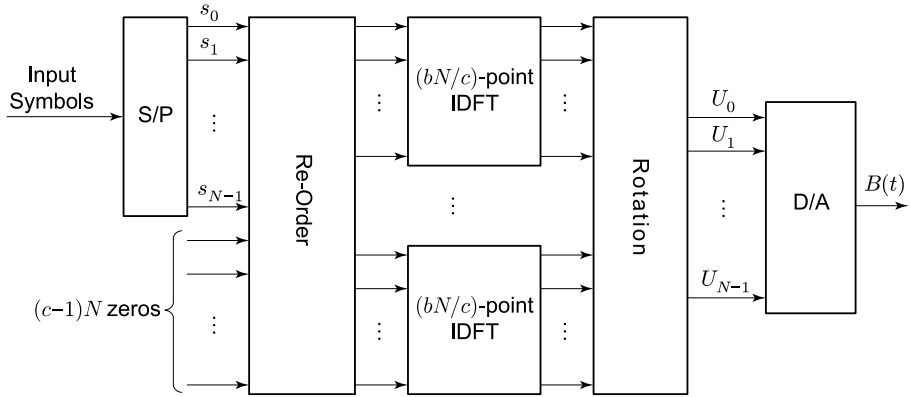}
\caption{The IDFT implementation for SEFDM transmitter. } 
\label{fig:red_tx} 
\end{figure}

From Theorem \ref{thm:sefdm} it can be shown that adding $c$ 
rotated OFDM systems can create the same signal as an SEFDM system.
This can be utilised to build an SEFDM transmitter as illustrated 
in Fig. \ref{fig:red_tx}. 
The transmitter starts 
by reordering the input symbols and insert zeros at appropriate 
locations to generate the $c$ symbol vectors. 
The symbols 
reorder block generates the $\bS'(k)$ vectors. 
These vectors are 
then fed into the $c$ IDFT modules. The outputs of the IDFTs are then 
rotated using the rotation matrices $\bR(k)$ and combined to generate the time 
sampled sequence $\bU$, which can  be fed into a digital to 
analogue converter (D/A) to finally
generate the continuous time signal $B(t)$. 

Generation of SEFDM with the IDFT offers many advantages. 
The IDFT based system is ready for digital implementation 
providing all the digital over analogue advantages. 
The structure of the SEFDM system is based on similar building blocks 
to the widely available OFDM system which will facilitate a smooth 
changeover. In addition, as will be shown 
later the receiver and transmitter have the same structure 
which can enable dual operation of the same equipment and consequently 
reduce the design and implementation costs.

\subsection{Receiver/decoder design}
\label{sec:decoder}

Once the SEFDM signals for one symbol period have been received they must
then be decoded to return the original symbol. 
The receiver therefore attempts to recover the original symbols by
decoding the interleaved OFDM systems individually by subtracting the
estimated interference from the other OFDM systems (for this reason we term this
the ``stripe" decoder).
Note that the design here is heuristic, no
proof of convergence is given (and one may not be possible).  The
justification for the design is that it is intuitive and it works
in software tests.  The receiver/decoder is shown diagrammatically as a 
data flow diagram in Fig. \ref{fig:decoder}.

\begin{figure}
\centering
\includegraphics [width=7cm] {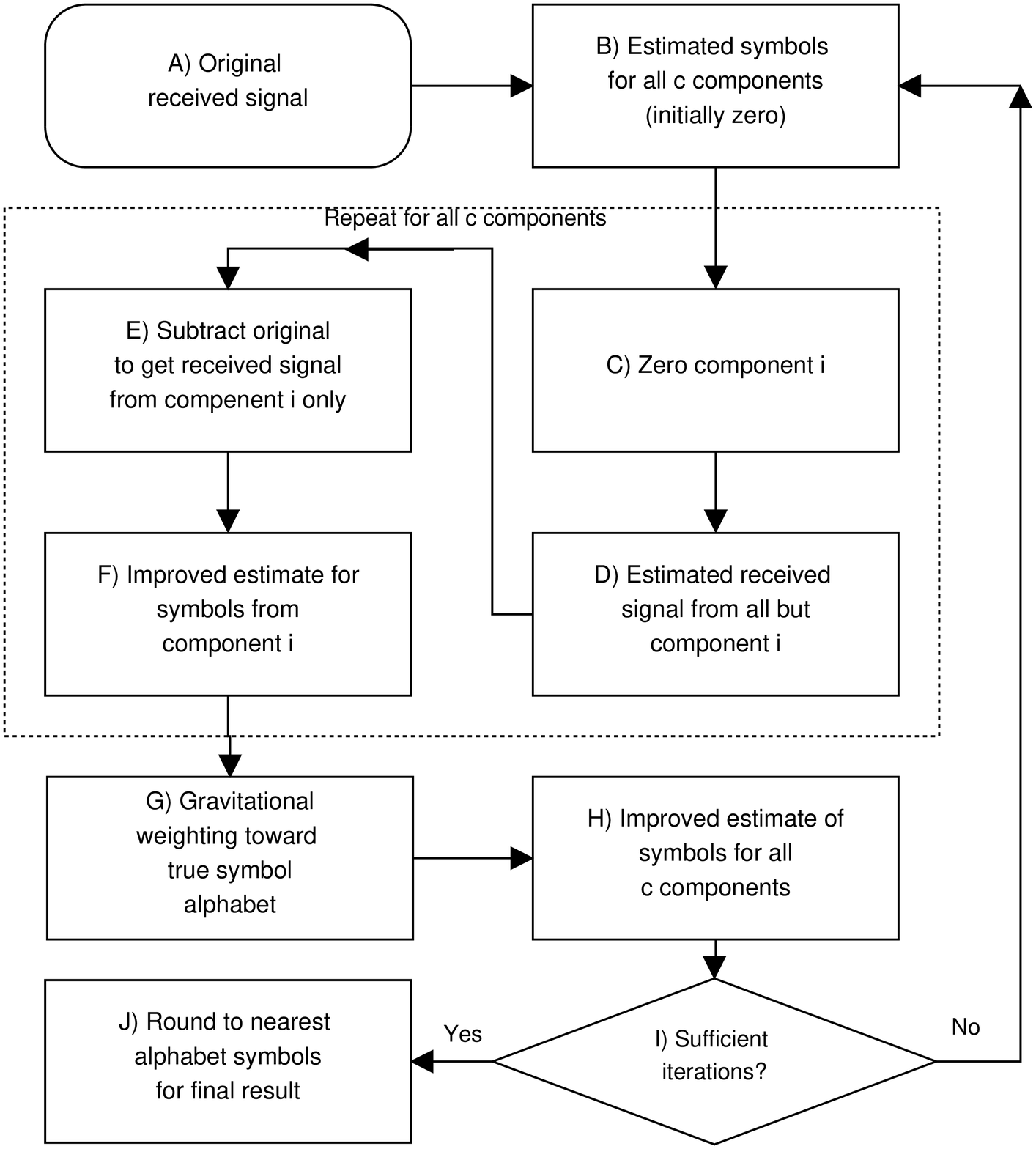}
\caption{Receiver/decoder dataflow diagram.} 
\label{fig:decoder} 
\end{figure}

Begin with a received signal (box A in diagram) and
an initial estimate that all symbols are
$0 + 0i$ (box B in diagram) and iteratively
improve the estimates by isolating the signal arising from each of
sub OFDM systems (see Fig. \ref{fig:sefdmexample}).  After several
iterations the estimates converge to the correct input symbol
and are eventually rounded to the closest symbol in the symbol
alphabet in use.

Consider, again, the
SEFDM system with $N$ carriers and $\alpha=b/c$.
Let $\bU$ (as in section \ref{sec:proof})
be the received signal (for now assume it is not corrupted by noise).  
If the system is OFDM, decoding is simple.  The
received frequencies are orthogonal and a simple IDFT recovers the
symbols on each carrier.  
Now, it follows that if the symbols for $c-1$ of the interleaved OFDM
systems were known then the symbols
on the remaining OFDM system could be obtained.  This is achieved
by firstly subtracting that portion of $\bU$ which arises from the $c-1$
OFDM systems with known symbols and secondly, performing the inverse DFT.
Using the notation of section \ref{sec:proof}, if $\bU(k)$ is the
signal arising from the $k$th interleaved OFDM system then 
$\bU - \sum_{k=1}^{c-1} \bU(k)$ is the signal arising from the zeroth
OFDM system $\bU(0)$.  An IDFT of $\bU(0)$ recovers the
symbols.  A similar process would be required if 
$\bU(0), \bU(2), \bU(3), \ldots, \bU(c-1)$ were known and $\bU(1)$
were to be recovered.  In that case a frequency shift $\bR(1)$ (again as
in section \ref{sec:proof}) would need to be applied before the inverse
DFT.  
It should be noted that even if $\bU$ is corrupted by AWGN, 
the above process can be used to produce a maximum likelihood
estimate of the original broadcast symbols by
rounding it to the nearest letter in the
``symbol alphabet" being used.

Given estimates of
the correct symbols for the $c$ interleaved OFDM systems then
improved estimates can be produced (inner dotted box in
diagram).  The estimates are produced
by, for each OFDM system in turn: first subtract the signal from
the $c-1$ other OFDM systems (box C, D and E in diagram) 
and secondly perform an inverse DFT 
with frequency shift to get an improved estimate for that OFDM subsystem
(box F in diagram).   To improve performance a
``gravitational" model is added to pull estimates towards symbols
in the symbol alphabet (box G in diagram).
This is repeated for $J$ iterations (box I in diagram).
Note that estimates are ``soft estimates" (complex numbers which
are not necessarily 
members of the symbol alphabet) until the final stage of processing
the estimates are mapped to the nearest member of the symbol
alphabet (box J in diagram).

\begin{enumerate}
\item Set $\widehat{\bS}$ the estimated symbols to $0+0i$ (box B).
\item Let $j:=1$ ($j$ counts iterations -- there are $J$ iterations).
\item Let $\widehat{\bS(0)}$, $\widehat{\bS(1)}$, $\ldots$ be
the estimates for the symbols of the $c$ interleaved OFDM systems.
The $\widehat{\bS(k)}$ together
make $\widehat{\bS}$ as in section \ref{sec:sefdm_intro}.
\begin{enumerate}
\item For each of the $c$ systems in turn, remove that part of
the signal generated by all symbols in $\widehat{\bS}$ apart from
$\widehat{\bS(k)}$ (box C and D).  Use this to estimate a new $\widehat{\bS(k)}$
and hence a new $\widehat{\bS}$ (box E and F). 
\item For each of the $N$ symbols calculate a ``gravitationally weighted"
version of $\widehat{\bS}$, $G(\widehat{\bS})$ (box G).
\end{enumerate}
\item $\widehat{\bS}:=  \widehat{\bS}(J-j)/J + (j/J) G(\widehat{\bS})$ (box  H).
\item If $j < J$ then $j:=j+1$ and go to step 3 (box I).
\item Finally, $\widehat{\bS}$ is ``sliced" to the nearest alphabet symbol
for each estimated symbol $\widehat{S_i}$ (box J).
\end{enumerate}

The two central steps of the algorithm (a) and (b) above
require slightly more explanation.  If $\br$ is the received signal
corrupted by noise then, to estimate the $k$th OFDM system first
calculate $C(k) = \br - \sum_{j=0, j \neq k}^{c-1} \widehat{\bU(j)}$ where
$\widehat{\bU(j)}$ is the estimated signal transmitted by just the estimated
symbols in
the $j$th OFDM system $\widehat{\bS(j)}$ (box D).  
$C(k)$ can then be shifted in frequency by 
multiplication $\bR(k)$ to produce an estimate of the signal (plus noise)
arising solely from the $k$th OFDM system (box E).  This
can be decoded in the usual OFDM manner (using IDFT) as
if it were an OFDM system transmitting
on every $b$th carrier.  This produces a new estimate for the symbols
on the $k$th OFDM system $\widehat{\bU(k)}$.  

These estimate symbols are truncated
to ensure that no symbol has a real or imaginary part outside the
range of the signalling alphabet (for example, if the system is 4-QAM estimates
are rounded so all real and imaginary parts are in the range $[-1,1]$.
The new estimate for $\widehat{\bU(j)}$ 
can immediately be used to update $\widehat{\bS}$ (box F).  This takes place
for each of the $c$ OFDM systems in turn (dotted large box) 
to produce a new $\widehat{\bS}$.

Note that while this part of the decoder design seems complicated, in
fact, the decoder can be implemented using the transmitter.  To calculate
$C(k)$ from $\br$ the received signal and $\widehat{\bU(j)}$ (the estimated
symbols on all carriers $j \neq k$) simply feed the estimated symbol set
$\widehat{\bS}$
with symbols $k, k+c, k+2c,\ldots$ set to zero to the transmitter.  This produces
an estimate of the signal which would be transmitted by all but the $k$th
OFDM system (corrupted by noise).  This signal can be decoded using IDFT
as in standard OFDM to produce an improved estimate for $\widehat{\bS(k)}$
the symbols of the $k$th OFDM system.

The ``gravitationally weighted" $G(\widehat{\bS})$ (box G) is calculated by
examining each estimated symbol in turn $\widehat{S_0}$, $\widehat{S_1}$ and so
on and producing the weighted sum of each of the symbols in the alphabet
weighted by the inverse of the distance to them (as a gravity law).  
If $A$ is the symbol alphabet
(say, $1+0i$ $-1+0i$ for BPSK)
then 
$$G(\widehat{S_i}) = K \sum_{a \in A} a/d(a,S_i)^2,$$
where $d(a,b)$ is the Euclidean distance between the two points in
complex space and $K$ is the normalising constant $1/\sum_{a \in A} a/d(a,S_i)^2$.
If for any $a \in A$, $d(a,S_i) =0$ then $G(S_i) = a$, that is, if the estimated
symbol happens to be exactly on a point in the alphabet (to machine precision)
then that point is returned.  Many similar weighting schemes could be tried but
this one appears sufficiently effective and quick to calculate.

The complexity of the decoder system is tied to the complexity of the transmitter.
(The ``gravitational" part is of negligible complexity).  To subtract an estimated
signal a signal is generated by the transmitter.  The final complexity of the decoder
then is a fixed linear multiple of the transmitter complexity -- this multiple being
a product of $c$ (the number of interleaved OFDM systems) and $J$, the number of iterations
in step (a) above (experiment found 20 to be a reasonable value).

\subsection{Comments on Implementation}

Numerical results modelling the work reported here demonstrate attractive error
performance (as will be seen in the next section)
at a much reduced complexity when compared to optimal
iterative detection algorithms such as SD. Ultimately the aim is to
realize the proposed designs in hardware. Examining the structure
of the proposed system reveals the support of an efficient implementation
path. The transmitter design relies on general purpose IDFT operations
which can be efficiently evaluated with the Inverse Fast Fourier Transform
(IFFT). We have recently reported the implementation of such transmitter on a 
reconfigurable field programmable gate array (FPGA)
architecture \cite{marcus_2011_ict} and demonstrated its operation,
showing its ability to perform real time tuning of $\alpha$.  Furthermore, 
design studies as very large scale integrated circuit (VLSI) 
structures have also been reported in \cite{safa201105_ISCAS}.

Examining the structure of the stripe decoder, shows that the main
components are standard DFT modules which can in turn be realized
with the FFT algorithm. Implementations of DFT based demodulators
for SEFDM system have also been reported in \cite{safa_pimrc2011}. 
A main difference in the design is that multiple DFT blocks are needed for
the stripe decoder while a single longer DFT is implemented in \cite{safa_pimrc2011}.
However, the DFT blocks arrangements in the stripe decoder may follow
the same pattern as those of the transmitter design.

\section{Simulation results}
\label{sec:results}
Numerical simulation was carried out to determine the performance of the transmitter
and decoder
system (as mentioned in the introduction the transmitter has been tested in hardware
this work is ongoing for the receiver/decoder).  
A sampled SEFDM system is characterised by $N$ (the number of carriers
in one symbol period), the symbol alphabet (what allowable complex
symbols are considered), $M$ (the number of samples obtained for decoding
in one symbol period), $\alpha=b/c$ (the compression ratio) and $\ebno$ the energy 
per bit to noise power spectral density 
ratio.  
As previously remarked, the spectral efficiency of an SEFDM system compared
with OFDM is $1/\alpha$.  So, for $\alpha = 5/6$ the spectral efficiency increases
by 20\% and for $\alpha=4/5$ by 25\%.

\subsection{Simulation description}

Code to simulate the system was written in python.  The code implements
transmission of SEFDM using the FFT method 
described in section \ref{sec:design}.  Test signals are generated from random bits.
Additive  White Gaussian Noise (AWGN) with
a given $\ebno$ is then added.  Channel effects such as fading 
and frequency and phase offsets
as well as system aspects such as 
channel estimation are not
considered in this work. The assumption of a simplified AWGN channel 
serves to illustrate the
basic concepts of the work and its practicability. More sophisticated channel 
models are the subject of ongoing work. It has been shown that joint detection 
and equalisation using sphere decoder can provide attractive BER performance \cite{ersi2010}. 
The authors believe that
a joint detection and equalisation technique based on the proposed
detection algorithm from this paper could be used to alleviate the problem.

The simulations described here are all performed with the assumption that data
is arriving as fast as the system can broadcast it (that is, the system
is at maximum load and there is always
a symbol on every channel in every period).  The results would not be
altered if this load fell (a blank symbol could be broadcast).
The choice of symbols is
random.  As the relationship of symbol patterns to BER is
unknown, completely random choices of symbols is the best way to 
obtain the actual BER a working system would have.

Three decoding schemes are implemented.  The first is
the ``stripe" decoding technique from this paper
(section \ref{sec:decoder}). 
The maximum likelihood (ML) method explicitly
tests every possible combination of alphabet symbols on each carrier and measures
the difference between the time series generated and the received signal.  While
this is in some sense ``perfect" as a decoding scheme it is computationally 
intractable for large $N$ (assuming for simplicity that $M=N$).  The number of
tests requires increases as $O(A^N)$ where $A$ is the number of symbols in
the alphabet and $N$ the number of carriers and assuming each test can be
performed with FFTs in parallel then each is of order $O(N \log N)$.
By contrast the stripe decoder complexity is $O(N \log N)$ (although this
must be multiplied by the constant $J$ the number of iterations performed).
The ``sphere decoder" method
attempts to more intelligently assess only the ``nearly correct" 
symbol sequences.  That is it uses a dynamic programming technique 
over only a subset of the possible symbol space.  However, because
it relies on numerical matrix inversion, it suffers problems with large
numbers of channels or low values of $\ebno$ as the number of symbol
combinations investigated becomes large.
The three methods are referred to in the results as ML, stripe and sphere.
The ML should represent a ``best possible" result and the sphere decoder
should also be optimal except in cases where the algorithm fails
to find a solution -- in practice the result coincides almost exactly with
the optimal solution where that is known.

To get statistically representative results, a high number of iterations
must be performed with each iteration representing one symbol period.
95\% confidence intervals have been calculated for
all experiments which measure bit error rate
on the assumption that each decoded bit is an independent 
trial (in fact, for say a 128 carrier 4-QAM system the error rates on 
groups of 256 bits composing one symbol period
will be loosely correlated but between simulated symbol periods the
bits are independent trials).  For most of the graphs plotted the 95\%
confidence intervals are too close together to show up and are omitted.

For space reasons runtime efficiency results are not shown here (and
results of runtime for computer simulation are not expected to translate
directly to better performance when implemented on hardware). 
The runtime results for transmitter and decoder 
were completely in line with the expected theoretical
results -- the time taken to produce a signal for one symbol period 
using the transmitter code was $O(M \log M c)$ where $c$ is the number 
of OFDM systems to be added and $M$ the number of samples per symbol
period.  To get accurate experimental estimates for BER it was necessary
to generate and decode tens of thousands of symbol periods (since the
BER was extremely low).  In our software simulations the ``stripe"
decoder could transmit and decode 128 symbols per period in an
acceptable runtime whereas the ML decoder could perform no more than 4 and
the sphere decoder no more than 8.  In short, the transmitter design and
receiver/decoder designs were, as predicted, a fixed, small multiple of
the runtime of an FFT routine.

\subsection{Decoder results, prediction accuracy}

It should be emphasised throughout this section that the OFDM result
(the theoretical BER line) represents the best possible result obtainable
in the case of an orthogonal system.  The maximum likelihood
(ML) result represents (within the bounds of statistical errors)
a best possible result for the simulation parameters used ($N$, $M$, $\alpha$ and
$\ebno$).  The sphere decoding result will also be ``near" optimal for
the simulation parameters used.


\begin{figure}
\centering
\includegraphics[width=8cm]{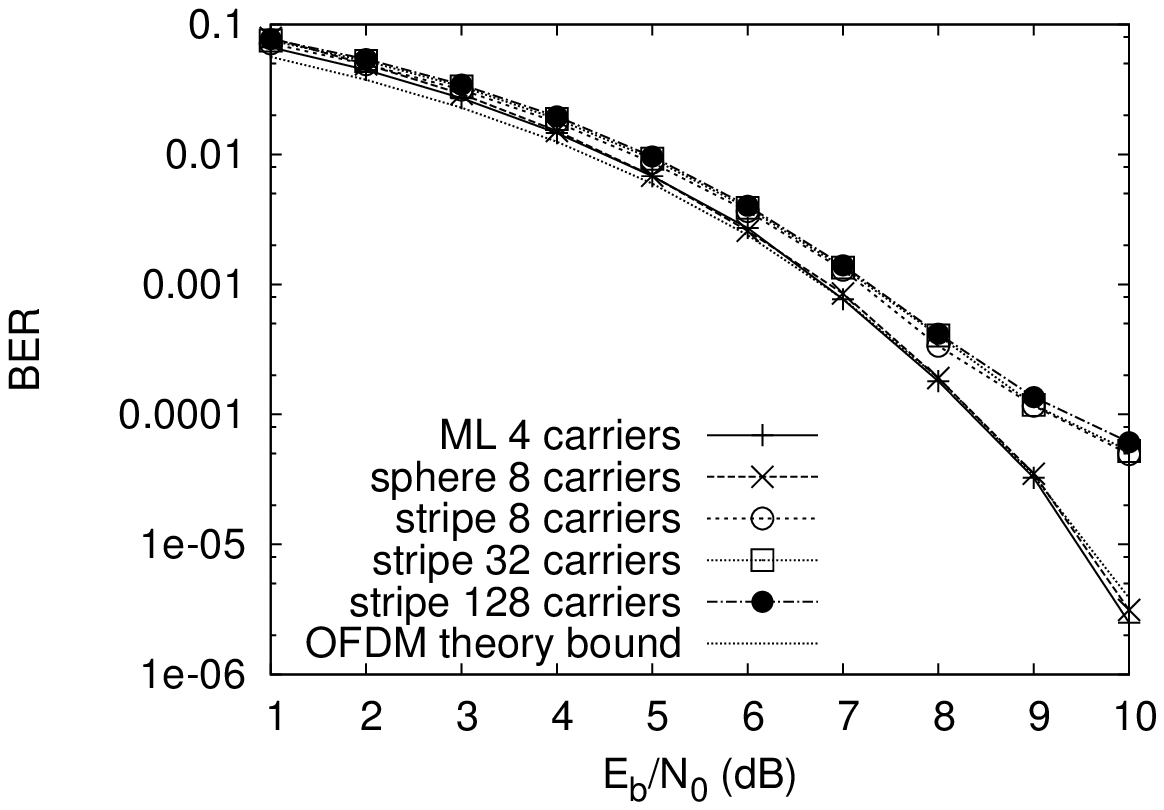}
\includegraphics[width=8cm]{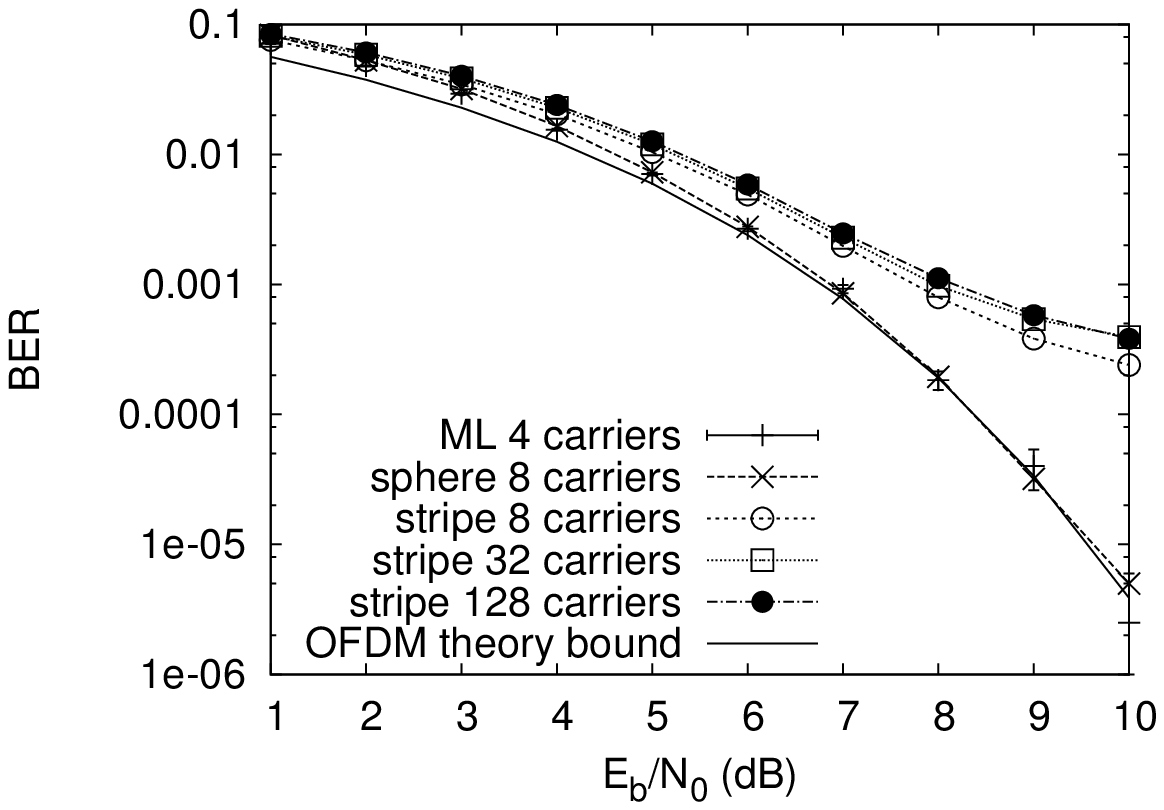}
\caption{Bit error rate for $\alpha=5/6$ (left) and
for $\alpha=4/5$ (right)
using 4-QAM.}
\label{fig:4QAM_a4_5a5_6}
\end{figure}

Fig. \ref{fig:4QAM_a4_5a5_6} (left)  shows results for the stripe decoder using
4-QAM for $\alpha = 5/6$.  The sphere decoder and the ML decoder
results are a good match with the theoretical
best possible except for low $\ebno$ (where they do not quite meet
the optimal bound as we might expect).  However, it is worth reiterating
that the theory applies to an idealised situation with complete knowledge
of the whole time signal in analysis but the simulation (and a real working
system) only considers samples.  The stripe
decoder certainly shows worse performance than the perfect theoretical
performance.  However, it is comparable to the sphere decoder in low $\ebno$
and for $\ebno > 5.0$dB the worsening of performance is equivalent to
approximately an extra 1 dB of $\ebno$
(the BER for OFDM at $9.0$dB  is the same as
is the BER for SEFDM at $8.0$dB).  
Note that it is this horizontal separation
which is relevant since the design question is ``how much more power (or less noise)
would be necessary to regain the lost performance?"
This is certainly very good performance.  In low 
$\ebno$ environments the stripe decoder would certainly be as good in terms
of BER as OFDM.  In high $\ebno$ environments (above 8dB) the stripe 
decoder is able to achieve an acceptable bit error rate for wireless systems 
where BER below 0.0001 are entirely reasonable although an OFDM system
would have lower errors.  In this case an SEFDM could either fall
back to OFDM or transmit at a higher power to reduce the $\ebno$
until the BER was acceptable.

Fig. \ref{fig:4QAM_a4_5a5_6} (right) shows results for the stripe decoder using
4-QAM for $\alpha = 4/5$ -- this is just below the limit of $\alpha = 0.802$
which is considered the ``best possible" for idealised recovery of the
signal.  The stripe
decoder again shows degraded performance.  
However, it is again comparable to the sphere decoder in low $\ebno$
and for $5.0$dB $< \ebno < 9.0$dB the worsening of performance is equivalent 
to an extra
1 or 1.5 dB of noise (that is the BER for SEFDM at $9.0$ dB is the same as
the BER for OFDM at $7.5$ dB).  
This is an acceptable power penalty/price to pay  given the advantage of 
bandwidth saving.  In low 
$\ebno$ environments the stripe decoder would certainly be as good in terms
of BER and much preferable in terms of spectral efficiency.  However, it
seems that the performance has worsened by going below the theoretical
$\alpha = 0.802$ limit even slightly.




\begin{figure}
\centering
\includegraphics[width=8cm]{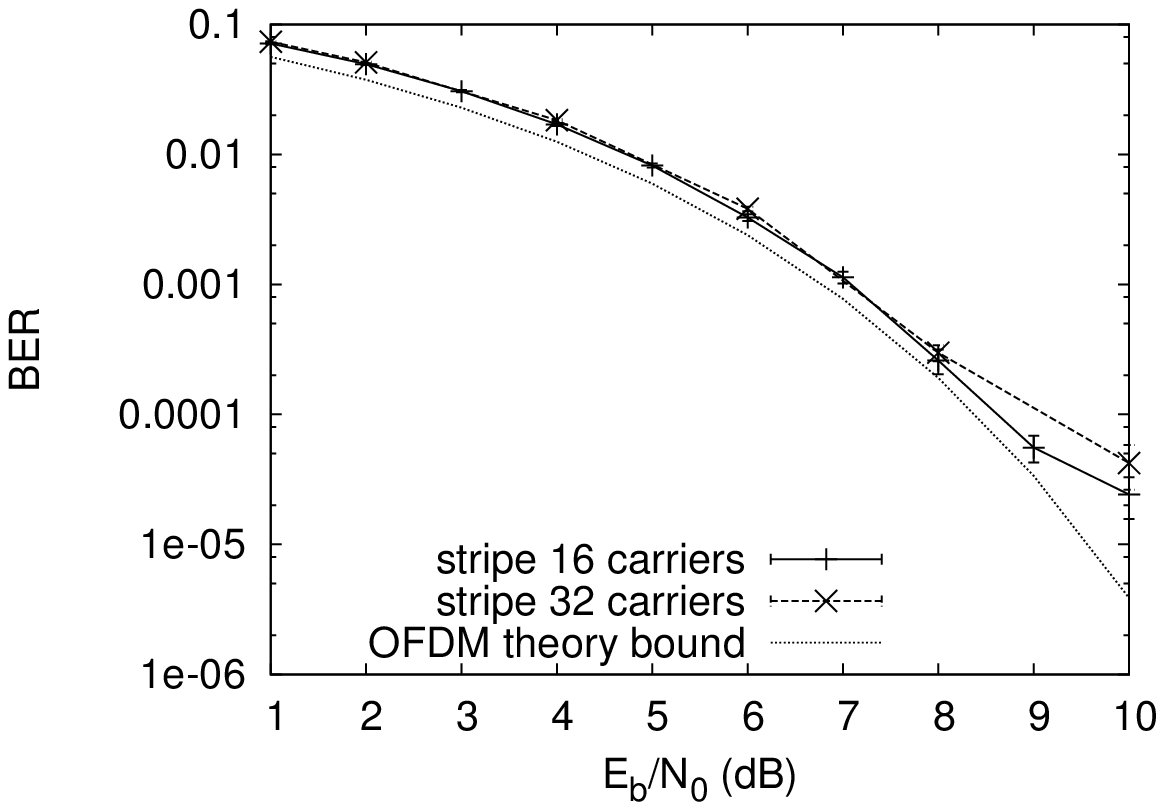}
\includegraphics[width=8cm]{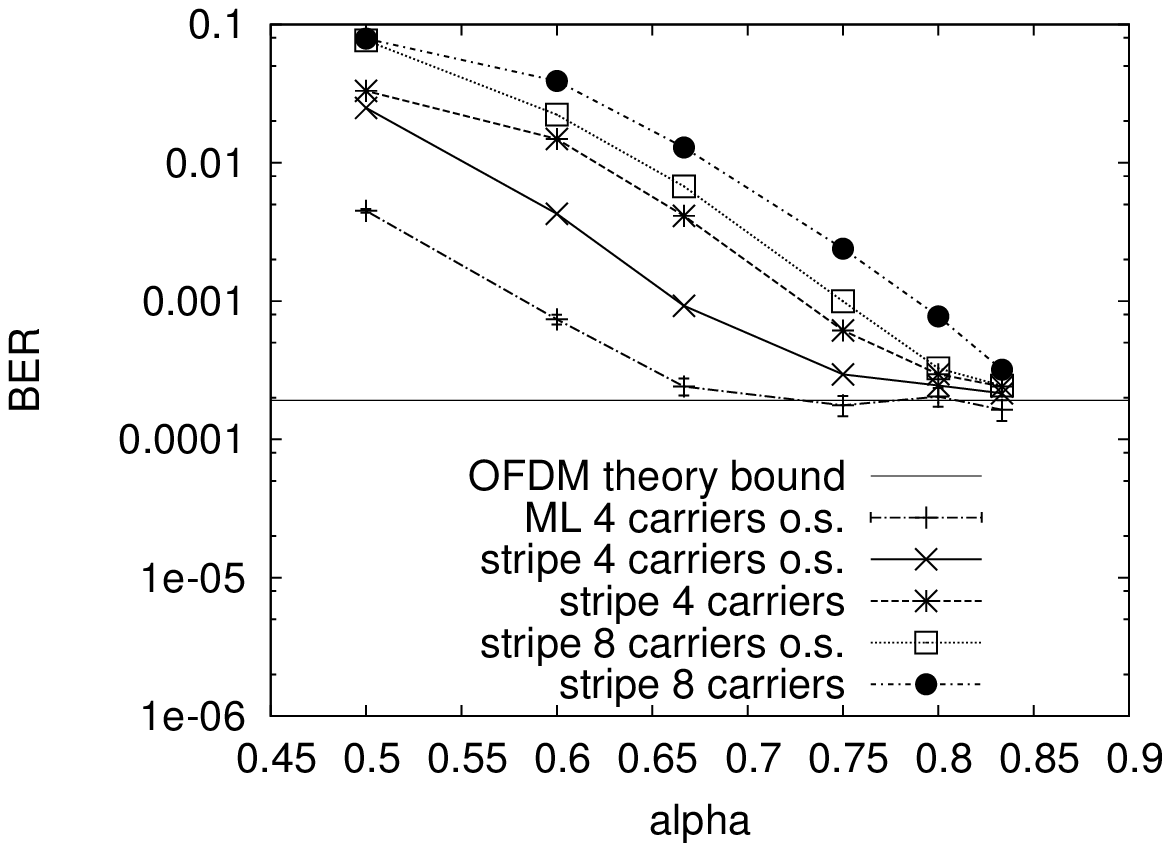}
\caption{Bit error rate for 4-QAM with oversampling, $\alpha=5/6$
fixed, $\ebno$ varying (left) and $\ebno = 8.0$ dB fixed $\alpha$ varying (right).}
\label{fig:4QAM_os}
\end{figure}

Fig. \ref{fig:4QAM_os} shows the improvements which oversampling
can bring.  
Recall from section \ref{sec:sefdm_intro} that the simulation in fact over
estimates interference when the number of samples is ``low".  More
samples will produce interference levels closer to the real life 
(continuous) system.

Fig. \ref{fig:4QAM_os}(left) shows the results for 4-QAM with
oversampling with $\alpha = 5/6$.  For 16 and 32 carriers an oversampling rate such
that $M=16N$ is tried --- 16 samples per
carrier in every symbol period.  This should be compared with Fig. 
\ref{fig:4QAM_a4_5a5_6} (left) which is the result for $N=M$ --- 
one sample per carrier in every
symbol period.
For 16 and 32 carriers the BER has almost no worsening from the theory except in the
highest signal to noise ratio where the degradation still remains modest.
Overall, then the performance of
the algorithm is extremely satisfactory with oversampling.
Oversampling
results for $\alpha=4/5$ are less successful however. 

Fig. 
\ref{fig:4QAM_os}(right) shows the results with 
$\ebno=8.0$ dB and $\alpha$ varied. 
The stripe detector is tried with 4 and 8 carrier systems and
heavy over sampling --- 
in this experiment
$M=128$ when $N=4$ or $N=8$.  In this case the improvement is
marked with a significant improvement in BER using oversampling.  Since the
stripe detector can perfectly happily function with large channel numbers it
can also work well with smaller numbers of channels and oversampling.  
The oversampling also improves the performance of the ML estimator, making it 
stay closer to the theoretical OFDM limit for smaller values of $\alpha$.
This figure, however, shows an important theoretical limit to what can
be gained by SEFDM type systems even with optimal detection and an extremely
small number of channels.  When $\alpha < 2/3$
the BER begins to increase markedly.  Therefore, even with perfect detection
it could never be expected that spectral efficiency gains of more than
50\%  $(1/\alpha - 1 $ with $\alpha= 2/3$) can be achieved even for the four carrier
system.  For more carriers the limit of $\alpha \simeq 0.802$ seems likely to apply.
These oversampling results confirm the intuition from section \ref{sec:sefdm_intro}
that the discrete simulation exaggerates the interference effect and more samples
will bring the interference (and hence BER) down.

Finally, Fig. \ref{fig:BPSK} shows results using the stripe decoder on
BPSK for $\alpha = 1/2$ (with $\alpha=1/2$ the system is that of FOFDM 
\cite{Rodrigues_FOFDM}).
The results are nearly ``perfect" with little deviation from the theory line
for OFDM which represents the best possible BER for an OFDM system
with that $\ebno$.  The
decoding using the ``stripe" method is ideal for this scenario.  More
than 128 carriers cannot be tested quickly enough to get sufficiently
accurate error prediction for the lower $\ebno$ values.  However, there
seems no reason to believe that the bit error rate increases with the
number of carriers in this scenario.  It can be seen that the results
are near ``perfect" for BPSK with $\alpha=1/2$.  
However, this is not as useful as it might appear since BPSK with 
$\alpha=1/2$ only carries the same amount of data as 4-QAM OFDM.

\begin{figure}
\centering
\includegraphics[width=8cm]{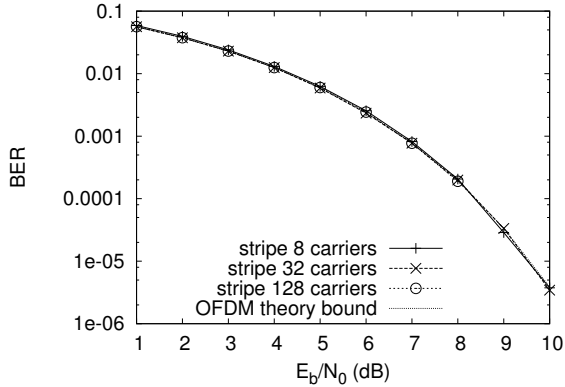}
\caption{Bit error rate using BPSK, fixed $\alpha= 1/2$ varying
$\ebno$ }
\label{fig:BPSK}
\end{figure}

In summary then, the results in this section show that the transmitter
and decoder perform well for 4-QAM with $\alpha=5/6$ but these good results
fall off to be a less acceptable performance for $\alpha=4/5$ in
tune with the expectation from the theoretical results of Rusek et al
\cite{rusek2005,Rusek_FasterThanNyquist} suggesting a lowest possible
value of $\alpha = 0.802$ before interference cannot be compensated
for.  The system performed better with heavy over sampling as
suggested by section \ref{sec:sefdm_intro} and performed extremely well 
(indeed the results were indistinguishable from optimal) for BPSK 
with $\alpha=1/2$.
Although no channel model was used in this simulation, complementary work
in \cite{ersi2010} shows
that SEFDM detection and equalisation can give good BER performance
in dispersive channel environments when the receiver employs a 
regularised sphere detection mechanism.

It remains to be seen whether the system would be practical for modern 
systems with a very large number of carriers (512 and beyond). 
Detailed investigation of the properties of SEFDM in \cite{safa_thesis}
have shown that the condition number of the matrix representing the
carriers  increases with the number of carriers and this increases the
complexity of the problem for any detection method which involves matrix
inversion. However, with the detection method proposed here, 
the error rate is not expected to be seriously compromised and 
we are encouraged by the results shown in Fig. \ref{fig:4QAM_a4_5a5_6} 
where there is only a slight degradation of the error when the 
number of carriers is increased from 16 to 128. Current software 
limitations for testing 
with larger number are being addressed by implementing the transmitter
and receiver in hardware and this is underway.  The building block
for the system is the DFT as with ODFM and, hence, the speed
of execution is not expected to be an issue
in a hardware implementation.

\section{Conclusions}
\label{sec:conclusions}
This paper describes the design of a simple system for transmitting, receiving
and decoding Spectrally Efficient FDM (SEFDM) signals.  These signals can
simply be generated by a transmitter mechanism very similar to that of 
standard OFDM with little increase in complexity.  The decoder is more difficult to
implement but the increase in complexity with the number of channels remains that of
standard OFDM $O(M \log M)$ (where $M$ is the number of samples).

Detailed modelling and simulations show that the decoder described in this paper can give 
an increase in spectral efficiency of 20\% ($\alpha=5/6$) with little noise
penalty and even 25\% 
($\alpha=4/5$) in some circumstances (with a noise penalty close to 2dB 
for a BER of $10^{-4}$).
Oversampling can be used to compensate
for almost all of the noise penalty for $\alpha=5/6$.  With oversampling this
system is ``almost perfect" producing the expected gain in spectral
efficiency, relative to OFDM, with minimal error degradation.

Naturally, work remains to be done in this area.  The decoder proposed here is a simple
heuristic chosen because it gets a ``good enough" solution in a very short time. 
It seems
likely that similar heuristics could close much of the small gap between 
the solution here and the ``optimal" solution.  The simulations here do not account
for channel fading, however, modelling using techniques similar to those
of \cite{ersi2010} to
perform channel equalisation in SEFDM are currently underway.

In conclusion, the system proposed and modelled here could produce gains in spectral efficiency
compared with an equivalent OFDM system.  The system is only slightly more complex to
implement than the OFDM system and functions in environments with similar noise levels,
particularly when oversampling is used.

\bibliographystyle{IEEEtran}
\bibliography{All}

\end{document}